\documentclass[]{interact}

\usepackage{epstopdf}
\usepackage[caption=false]{subfig}

\usepackage[numbers,sort&compress]{natbib}
\bibpunct[, ]{[}{]}{,}{n}{,}{,}

\theoremstyle{plain}

\theoremstyle{definition}

\theoremstyle{remark}
\newtheorem{remark}{Remark}

\usepackage{dsfont}                  
\DeclareMathOperator{\R}{\mathds{R}} 
\usepackage{mathtools}               
\usepackage{thm-restate}             
\usepackage{siunitx} 
\usepackage{xspace}                  
\usepackage{array}                   
\usepackage{multirow}
\usepackage{threeparttable}
\usepackage{etoolbox}                

\appto\TPTnoteSettings{\figcaptionfont}

\newtheorem{assumption}{Assumption}
\DeclareMathOperator\indicator{I}
\DeclareMathOperator{\fp}{FP}
\DeclareMathOperator{\fn}{FN}
\DeclareMathOperator{\tp}{TP}
\DeclareMathOperator{\tn}{TN}
\DeclareMathOperator{\fdr}{FDR}
\DeclareMathOperator{\for}{FOR}
\DeclareMathOperator{\X}{\mathbf X}
\DeclareMathOperator{\xe}{\mathbf x}

\DeclareMathOperator{\y}{\mathbf y}

\DeclareMathOperator\I{\mathbf {I}}
\DeclareMathOperator\J{\mathbf {J}}
\DeclareMathOperator\cD{\mathcal {D}}

\DeclareMathOperator\N{\mathcal {N}}

\DeclareMathOperator{\D}{\mathbf D}
\DeclareMathOperator{\bd}{\mathbf d}

\DeclareMathOperator{\defeq}{\vcentcolon=}

\DeclareMathOperator{\dist}{dist}

\DeclareMathOperator{\cor}{cor}

\DeclareMathOperator{\E}{E}

\DeclareMathOperator*{\argmax}{arg\,max}

\newcommand{\wu}[1]{\overset {_\sim}  #1}
\newcommand{\pc}{\ensuremath{\%\Delta}\xspace}

\newcolumntype{m}{>{$}c<{$}}
\newcolumntype{L}[1]{>{\raggedright\let\newline\\\arraybackslash\hspace{0pt}}b{#1}}
\newcolumntype{C}[1]{>{\centering\let\newline\\\arraybackslash\hspace{0pt}}b{#1}}
\newcolumntype{R}[1]{>{\raggedleft\let\newline\\\arraybackslash\hspace{0pt}}b{#1}}

\begin{document}


\title{A statistical testing procedure for validating class labels}

\author{
\name{Melissa C. Key\textsuperscript{a}\thanks{CONTACT Melissa Key. Email: melissa.key.ctr@afit.edu} and Benzion Boukai\textsuperscript{b}}
\affil{\textsuperscript{a}Department of Biostatistics, Fairbanks School of Public Health, Indianapolis, Indiana; \textsuperscript{b} Department of Mathematical Sciences, Indiana University-Purdue University, Indianapolis, Indiana}
}

\maketitle

\begin{abstract}
	
	Motivated by an open problem of validating protein identities  in label-free shotgun proteomics work-flows, we present a testing procedure to validate class/protein labels using available measurements across instances/peptides.
	More generally, we present a solution to the problem of identifying instances that are deemed, based on some distance (or quasi-distance) measure, as outliers relative to the subset of instances assigned to the same class.
	The  proposed procedure is non-parametric and requires no specific distributional assumption on the measured distances.
	The only assumption underlying the testing procedure is that measured distances between instances within the same class are stochastically smaller than measured distances between instances from different classes.
	The test is shown to simultaneously control the Type I and Type II error probabilities whilst also controlling the overall error probability of the repeated testing invoked in the validation procedure of initial class labeling.
	The theoretical results are supplemented with results from an extensive numerical study, simulating a typical setup for labeling validation in proteomics work-flow applications.
	These results illustrate the applicability and viability of our method.
	Even with up to \SI{25}{\percent} of instances mislabeled, our testing procedure maintains a high specificity and greatly reduces the proportion of mislabeled instances.
\end{abstract}


\section{Introduction}
The research presented in this paper is motivated by an open problem in the quantification of proteins in a label-free shotgun proteomics work-flow.
More generally, it presents a non-parametric solution to the problem of identifying instances that are outliers relative to the subset of instances assigned to the same class.
This serves as a proxy for finding errors in the data set: instances for which the class label is recorded incorrectly, or where the measurements for a particular instance are sufficiently inaccurate as to render them uninformative.

In label-free shotgun proteomics, the experimental units of interest (proteins) are not measured directly but are represented by measurements on one to 500+ enzymatically cleaved pieces known as peptides.
The amino acid sequences composing each peptide are not known apriori, but inferred based on algorithmic procedures acting on spectrum data from the mass spectrometer.
By removing inaccurate peptides (instances) from each protein (class), subsequent quantitative analyses which assume that all measurements are equally representative of the protein are thus more accurate and powerful.

A similar problem exists in classification theory.
It has been shown that classification models trained on data with labeling errors tend to be more complex and less accurate than models trained on data without labeling errors \citep{Quinlan1986,Zhu2004,Saez2014}.
Multiple algorithms to find and remove such data have been developed.
An excellent review of such procedures is presented in \citep{Frenay2014}.
While these procedures can be applied to the proteomics problem, it is important to note a distinction between the problem of interest and the one addressed by these algorithms.

In the general classification filtering problem, the aim is an improved ability to classify instances, and thus this is the primary criterion against which these algorithms are judged.
In other words, the accuracy of the filtered training set is secondary to the overall improved performance of the classification algorithm trained on this data.
In contrast, the overall accuracy of the filtered data set is of paramount interest in the proteomics filtering problem.
In a forthcoming paper, we will evaluate our proposed procedure against alternative algorithms from both proteomic and classification literature.
In this paper, our emphasis is on presenting the new testing procedure, and demonstrating its properties.

The paper is thus organized as follows.
Section \ref{sec:methodology} presents the theoretical basis for the algorithm as well as the procedure.
Section \ref{sec:simulation} uses a simulation study to demonstrate the effectiveness of the proposed algorithm.
Section \ref{sec:discussion} wraps up the paper with a discussion and concluding remarks.

\section{Methodology} \label{sec:methodology}
\subsection{The Basic Setup}

Consider a data set of $N$ instances, of which $ N_1,  N_2, \ldots, N_{K-1}$ are presumed to belong to classes $C_1,C_2,\ldots,C_{K-1}$ respectively,  with $ N_k > 1$ for $k = 1, \ldots, K-1$.  Let $C_K$ be a `mega-class' consisting of the $ N_K$ instances unassigned to a specific class or instances which are the sole representative of a class, such that $N_K = N - \sum_{k = 1}^{K-1}  N_k$.  For simplicity, we use the shorthand notation $i \in C_k$ to indicate an instance $i$ which belongs to class $C_k$ while $i \notin C_k$ indicates an instance which does not.  Let $\xe_i = (x_{i1},x_{i2},\ldots,x_{in})^\prime$ be the (vector) of observed intensities from instance $i$, (i = 1,\ldots,N), across $n$ (independent) samples.
The available data is thus $\X = (\xe_1,\xe_2,\ldots,\xe_N)$ is an $n \times N$ matrix of such observed intensities.
The `distance' between any two instances with observed intensities $\xe_i$ and $\xe_j$ can thus be measured using any standard distance or quasi-distance function,
\begin{equation}\label{eq:dist}
d_{ij} = \dist({\bf x}_i,{\bf x}_j) \ge 0.
\end{equation}

For instance, $d_{ij}$ could be a measure of the dissimilarity between peptides over the $n$ samples in the study.
In this case, one popular quasi-distance function is defined by the correlation between $\xe_i$ and $\xe_j$, $d_{ij}=1-r_{ij}$, where 
\begin{equation} \label{eq:cor-dist}
r_{ij}:= \cor(\xe_i, \xe_j)=\frac{\sum_{\ell = 1}^n \left(x_{i\ell} - \bar x_{i\cdot}\right)\left(x_{j\ell} - \bar x_{j\cdot}\right)}{\sqrt{\sum_{\ell = 1}^n \left(x_{i\ell} - \bar x_{i\cdot}\right)^2\sum_{\ell = 1}^n \left(x_{j\ell} - \bar x_{j\cdot}\right)^2}}.
\end{equation}
Here $\bar x_{i\cdot}=\sum_{\ell=1}^n x_{i\ell}/n$, for each $i=1, \dots, N$. 

Let $\D = \{d_{ij} : i,j = 1,\ldots N\}$ be the $N\times N$ (symmetric) matrix comprised of these between-instance observed distances.  Without loss of generality, we assume that the entries of $\D$ are ordered such that the first $ N_1$ entries belong to $C_1$, the next $ N_2$ entries belong to $C_2$, etc.
Accordingly, we partition $\D$ as
\begin{equation} \label{eq:D}
\D = \begin{bmatrix}
\D_{11} & \D_{12} & \ldots & \D_{1K}\\ 
\D_{21} & \D_{22} & \ldots & \D_{2K}\\
\vdots & \vdots & \ddots & \vdots\\
\D_{K1} & \D_{K2} & \ldots & \D_{KK}
\end{bmatrix},
\end{equation}
where $\D_{kk}$ is the $ N_k \times  N_k$ matrix of between-instance distances \emph{within} class $C_k$ and the elements of $\D_{k_1k_2}$  represent the distances between the $ N_{k_1}$ instances belonging to $C_{k_1}$ and the $ N_{k_2}$ instances belonging to $C_{k_2}$.
Note in particular that $\D_{k_1k_2}\equiv \D_{k_2k_1}^\prime$.

To begin with, consider at first class $C_1$ and the $N_1$ instances initially assigned to it.
For a fixed $i$, $i = 1, \ldots,  N_1$, let $i \in C_1$ be a given instance in class $C_1$ and set $\bd_i := (\bd_i^{(1)}, \bd_i^{(2)},\ldots,\bd_i^{(K)})$,  be the $i^{th}$ row of $\D$, where $\bd_i^{(1)} := (d_{i1}, \ldots, d_{i N_1})^\prime$ and for $k \ge 2$, $\bd_i^{(k)} := (d_{i(1 + \sum_{j=1}^{k-1} N_j)}, \ldots, d_{i(\sum_{j=1}^{k} N_j)})^\prime$.  
Clearly, $\bd_i^{(k)}$ is the $i^{th}$ row of $\D_{1k}$ for $k = 1,\ldots, K$.
For the class $C_1$, we consider the stochastic modeling of the elements of $\D_{11}, \D_{12}, \dots, \D_{1{K}}$.
We assume that for each instance  $i\in C_1$, the observed within-class distances from it to the other $N_1-1$ instances in $C_1$ are $i.i.d.$ random variable according to a class-specific distribution, $G(\cdot)$, so that  
\begin{equation}\label{eq:G_ind}
\left( d_{i1},\ldots,d_{i(i - 1)},d_{i(i + 1)},\ldots, d_{i N_1}\right) \underset{i.i.d.}{\sim} G_i(\cdot),
\end{equation}
(since $d_{ii} \equiv 0$).
Here, the c.d.fs $G_i(\cdot)$ are defined for each fixed $i\in C_1$, and any $j\in C_1$ as 
\begin{equation} \label{eq:G}
G_i(t) \equiv G(t):=\Pr(d_{ij}\leq t\, |\, i\in C_1, \ j\in C_1, \ j\neq i), \ \ \ \forall t \in {\R}.
\end{equation}
Our notation in (\ref{eq:G}) stresses that we are assuming that the distribution of distances from each individual instance to the remaining instances in $C_1$ is identical.
Similarly, the distances between the given instance, $i\in C_1$, and the $N_k$ instances in class $C_k, \ k=2, \dots, K$, are also $i.i.d.$ random variables according to some distribution $F^{(k)}(\cdot)$, so that   
\begin{equation} \label{eq:F_ind}
\left(d_{ij}\right) \underset{i.i.d.}{\sim} F^{(k)}(\cdot),\ \ \ \forall j \in C_{k},\ \ \ \ (j=1 \dots, N_k). 
\end{equation}
Here $F^{(2)}, F^{(3)}, \dots, F^{(K)}$ are $K-1$ distinct c.d.fs defined for each $i\in C_1$, and any $j\in C_k$ as
\begin{equation}\label{eq:Fk}
F^{(k)}(t)  = \Pr\left(d_{ij} \le t | \, i\in C_1, \ j\in C_k\right), \ \  \forall t \in {\R}.
\end{equation}
We assume throughout this work that $G(\cdot), F^{(2)}(\cdot), \ldots, F^{(K)}(\cdot)$ are continuous distributions with p.d.fs $g(\cdot), f^{(2)}(\cdot), \ldots, f^{(K)}(\cdot)$ respectively.
If we further assume that all the $ N_1$ instances from class $C_1$ are equally representative of the true intensity across all $n$ samples, we would expect that the distances between any two instances from within class $C_1$ are {\it stochastically smaller} than distances between instances from within $C_1$ and instances associated with class $C_{k}$ where $k \ne 1$.
Accordingly, we have
\begin{assumption}[\emph{Stochastic ordering}]\label{ass:F<G}
	For each $k, \ \ k=2, \dots, K$,
	\begin{equation}
	\Pr(d_{ij} \le t | \, i\in C_1, \ j\in C_k)\leq \Pr(d_{ij}\leq t\, |\, i\in C_1, \ j\in C_1, \ j\neq i)
	\end{equation}
	or equivalently,
	\begin{equation*}
	F^{(k)}(t) \le G(t)\ \ \forall t \in \R.
	\end{equation*}
\end{assumption}

In light of (\ref{eq:Fk}), the distribution of distances from the $i^{th}$ instance in $C_1$ to  any other random instance $J$, selected uniformly from among the $N-N_1$ instances not in $C_1$, is thus the mixture, 
\begin{equation} \label{eq:bar_F}
\bar F(t) := \Pr\left(d_{iJ} \le t |\, i \in C_1, J\notin C_1\, \right) = \frac{1}{N-N_1}\sum_{k = 2}^K N_{k} F^{(k)}(t),  
\end{equation}
where we have taken $\Pr(J\in C_k |\, J \notin C_1):=N_k/(N-N_1)$. Further, if $I$ is a randomly selected instance in class $C_1$, selected with probability $\Pr(I=i|I\in C_1)=1/N_1$, for $i=1, 2, \dots, N_1$, then it follows that
\begin{equation}\label{eq:F_bar}
\begin{aligned}[b]
\Pr (d_{IJ}\leq t| I\in C_1, \ J\notin C_1) & = \sum_{i=1}^{N_1} \frac{1}{N_1}\Pr (d_{iJ}\leq t| i\in C_1, \ J\notin C_1) \\
& = \sum_{i=1}^{N_1} \frac{1}{N_1} \bar F(t) \\
& \equiv \bar F(t) 
\end{aligned}
\end{equation}
Similarly, if $I$ and $J$ represent two distinct instances, both randomly selected from $C_1$ with $\Pr (I=i, \ J=j| I\in C_1, \ J\in C_1)=1/N_1(N_1-1)$, then 
\begin{equation} \label{eq:G_bar}
\begin{aligned}[b]
\bar G(t) &:= \Pr (d_{IJ}\leq t| I\in C_1, \ J\in C_1, \ J\neq I) \\
& = \sum_{i = 1}^{N_1}\sum_{j = 1, j \ne i}^{N_1} \Pr (d_{ij}\leq t|  i\in C_1, \ j\in C_1, \ j \ne i) \\
& = \sum_{i = 1}^{N_1}\sum_{j = 1, j \ne i}^{N_1} \frac{1}{N_1(N_1 -1)} G_i(t) \\
& \equiv G(t).
\end{aligned}
\end{equation}
It follows by Assumption \ref{ass:F<G} that 
\begin{equation}\label{eq:F<G}
\bar F(t)\leq \bar G(t), \  \forall t \in {\R}. 
\end{equation}

Now, for any $t \in \R$ define 
\begin{equation}\label{eq:psi}
\psi(t):= \bar G^{-1}(1-\bar F(t)). 
\end{equation}

It can be easily verified that the function $h(t):=\psi(t)-t$ has a uniqe solution, $t^*$,  such that $t^* = \psi(t^*)$ and 
\begin{equation}\label{eq:tau}
\bar G(t^*)  = 1 - \bar F(t^*) \defeq \tau.
\end{equation}
By Assumption \ref{ass:F<G} and (\ref{eq:F<G}), it follows that $\bar G(t^*) = \tau > 0.5$. 
As we will see below, the value of $t^*$ serves as a cut-off point to differentiate between the distribution governing distances between instances in class $C_1$ and the distribution of distances going from instances in $C_1$ to all remaining instances.

\subsection{The Testing Procedure}

\subsubsection{Constructing the test}

Consider at first class $C_1$ and the $N_1$ instances initially assigned to it.
Based on the available data we are interested in constructing a testing procedure for determining whether or not a given instance that was assigned to class $C_1$ should be retained or be removed from it (and potentially be reassigned to a different class). 
That is, for each selected instance from the list $i\in \{ 1, 2, \dots,  N_1\}$ of  instances labeled $C_1$, we consider the statistical test of the  hypothesis
\begin{equation} \label{H0i}
\mathcal H_0^{(i)}:   i \in C_1 \  \text{ (the initial label  is correct)}
\end{equation}
against
\begin{equation} \label{Hai}
\mathcal H_1^{(i)}:  i \notin C_1 \text{ (the initial label  is incorrect)},
\end{equation}
for $i=1, \dots, N_1$. 
The final result of these successive $N_1$ hypotheses tests  is  the set of all those instances in $C_1$ for which $\mathcal H_0^{(i)}:\ i\in C_1$ was rejected and thus, providing the set of those instances in $C_1$ which were deemed to have been mislabeled.
As we will see below, the successive testing procedure we propose is constructed so as to control the maximal probability of a type I error, while minimizing the probability of a type II error.

Towards that end, define  
\begin{equation} \label{eq:Zi}
Z_i\equiv\sum_{j\in C_1, \ j\neq i} \indicator[d_{ij}\leq t^*]=\sum_{j=1, \ j\neq i}^{ N_{1}} \indicator[d_{ij}\leq t^*]
\end{equation}
for each $i \in \{1, 2, \ldots, N_1\}$ where $t^*$ is defined by (\ref{eq:tau}) and $\indicator [{ A}]$ is the indicator function of the set ${ A}$.
In light of the relation (\ref{eq:F<G}), $Z_i$ will serve as a test statistic for the above hypotheses.
The distribution of $Z_i$ under both the null and alternative hypotheses can be explicitly defined as Binomial random variables, as is presented in the following lemma (proof  omitted).
\noindent \begin{restatable}{lemma}{nulldist}\label{lem:nulldist}
	\begin{itemize} Let $Z_i$ be as define in (\ref{eq:Zi}) above  with  $i=1, 2, \dots, N_1$,  then
		\item[a)] if $i\in C_1$, we have \ 
		$
		Z_i|_{\mathcal H_0^{(i)}}\sim Bin\left(N_1-1, \bar G(t^*)\right)\equiv Bin( N_1-1, \tau); 
		$
		
		\item[b)]  if $i\notin C_1$,  we have \
		$
		Z_i |_{\mathcal H_1^{(i)}} \sim Bin(\left(N_1 - 1, \bar F(t^*) \right) \equiv Bin(N_1 - 1, 1 - \tau). 
		$
	\end{itemize}
\end{restatable}
  
Accordingly, the statistical test we propose will reject the  null hypothesis $\mathcal H_0^{(i)}: i\in C_1$ in favor of $\mathcal H_1^{(i)}:\ i\notin C_1$ for small values of $Z_i$, say if $Z_i \leq a_{\alpha}$ for some suitable critical value $a_{\alpha}$ (to be explicitly determined below) which should satisfy, 
\begin{equation}\label{eq:alpha}
\hat \alpha:= \Pr\left(Z_i \leq a_{\alpha}\middle|\mathcal H_0^{(i)}\right)= \Pr\left(Z_i \leq a_{\alpha}\middle|\, \tau \right) \leq \alpha, 
\end{equation}
for each $i=1, \dots, N_1$ and some fixed (and small) statistical error level $\alpha \in (0,0.5)$.
The constant $a_{\alpha}$ is the (appropriately calculated)  $\alpha^{th}$ percentile of the $Bin( N_1-1, \tau)$  distribution.  
That is, if $b(k, n, p)$ denotes  the  c.d.f. of a $Bin(n, p)$ distribution, 
then for  given $\alpha$ and $\tau$, the value $a_{\alpha}$ is determined so as 
\begin{equation}\label{aa}
a_\alpha = \argmax_{k=0, \dots N_1-1}\big\{\! b(k,  N_1-1, \tau) \leq \alpha\big\}. 
\end{equation}

The final result of this repeated testing procedure is  given by the set of all instances in $C_1$ for which $\mathcal H_0^{(i)}: \ i\in C_1$ was rejected, 
\begin{equation*}
\mathcal R_{\alpha} := \{i;\ \  i=1, \dots,  N_1 : Z_i \leq a_{\alpha} \}, 
\end{equation*}
providing the set of those instances in $C_1$ for which the binomial threshold is achieved and therefore have been deemed mislabeled.
Similarly, 
\begin{equation*}
\mathcal A_{\alpha} := \{i;\ \  i=1, \dots,  N_1 : Z_i > a_{\alpha} \} = \{1, \ldots, N_1\} \setminus \mathcal R_{\alpha}, 
\end{equation*}
provides the set of instances correctly identified in $C_1$.
It remains only to determine the optimal value of $a_\alpha$ for the test.

\subsubsection{Controlling type I and type II  errors}\label{sec:typeIerrors}
With $\mathcal H_0^{(i)}$ and $\mathcal H_1^{(i)}$ as are given in (\ref{H0i}) and (\ref{Hai}), let 
\begin{equation}
\mathcal H^*_0= \bigcap_{i=1}^{ N_1} \mathcal H_0^{(i)}\  \  \text{and} \ \ \ \tilde{\mathcal H}^*_1= \bigcup_{i=1}^{ N_1} \mathcal H_1^{(i)}. \label{eq:bigH}
\end{equation}
The hypothesis $\mathcal H^*_0$ above states that all the instances in $C_1$ are correctly identified, whereas  $\mathcal H^*_1$ is the hypothesis  that at least one of the instances in $C_1$ is misidentified. We denote by $R=|\mathcal R_{\alpha}|$ the cardinality of the set $\mathcal R_{\alpha}$,  
\[
R=\sum_{i=1}^{ N_1}\indicator[ Z_i\leq  a_{\alpha}]
\]
so that $R$ is a random variable taking values over $\{0, \, 1, \dots,  N_1\}$.
Note trivially that $N_1-R\equiv |\mathcal A_{\alpha}|$. 

We consider the ``global'' test which rejects $\mathcal H_0^*$ in (\ref{eq:bigH})  if for at least  one $i, \ i=1, \dots,  N_1$, $ Z_i \leq a_{\alpha}$ or equivalently, if $\{R>0\}$.  The probability of a type I error  associated with this ``global'' test is therefore
\begin{equation} \label{eq:alpha0}
\begin{aligned}
\alpha': = &  \Pr\left(R >0\mid  \mathcal H_0^*\right) = \Pr\left(\bigcup_{i=1}^{ N_1}\{Z_i\leq a_{\alpha}\} \mid \mathcal H_0^*\right) \\
\leq & \sum_{i=1}^{N_1} \Pr\left(Z_i\leq a_{\alpha} \mid \mathcal H_0^{(i)}\right)= N_1\hat \alpha \leq N_1 \alpha, 
\end{aligned}
\end{equation}
using the Bonferroni inequality since by (\ref{eq:alpha})-(\ref{aa}),  $\hat \alpha\leq \alpha$.
The calculations for $\alpha^\prime$, can be controlled  by taking $\alpha=\alpha_0/N_1$ for some $\alpha_0$, to ensure that $\alpha^\prime \leq  \alpha_0$ and that $\hat \alpha\leq \alpha_0/N_1$. 

Note that if $\{Z_i, Z_2, \dots, Z_{ N_1}\}$ were to be independent or associated random variables \citep{Maxiumum1961} then under $\mathcal H_0^*$,  $\indicator[ Z_i\leq  a_{\alpha}\, ]\sim Bin(1, \hat \alpha), \ i=1, 2, \dots, N_1$ , and  $R\sim Bin(N_1, \hat \alpha)$.
In this case, 
\[
\Pr\left(R=0\, \mid\, \mathcal H_0^*\right)= \left(1-\hat \alpha\right)^{N_1}\geq \left(1-\frac{\alpha_0}{N_1}\right)^{N_1}.
\]
It follows for sufficiently large $N_1$ (as $ N_1\to \infty$), that
\begin{equation}\label{eq:ind-alpha}
\alpha'= 1- \Pr(R =0 \mid \mathcal H_0^*)\to 1-e^{-\alpha_0}<\alpha_0.
\end{equation}

The distribution of $Z_i$ under the alternative hypothesis is explicitly available (see Lemma \ref{lem:nulldist} (b)), so the type II error rate of the procedure can also be explicitly controlled.
In fact, the symmetry of the Binomal distribution about $\tau$ and $1 - \tau$, with $\tau>0.5$, can be exploited to show that when $\tau$ is sufficiently large ($\tau>\tau^*$, say), the type I error rate, $\alpha$,  serves also as a bound on the type II error rate, $\beta$, where we define
\begin{equation}
\beta := \Pr(\text{Type II error}) = \Pr(Z_i \ge a_\alpha | \mathcal H_1^{(i)}) .
\end{equation}
The conditions under which $\beta\leq  \alpha$ holds are provided  in the following lemma whose proof is given in the appendix below.
\begin{restatable}{lemma}{aControlsb}\label{lem:aControlsb}
	Suppose for a fixed $\alpha<0.5$, the test of the hypotheses $\mathcal H_0^{(i)}$ versus $\mathcal H_1^{(i)}$ as are given  in (\ref{H0i})-(\ref{Hai}) and is conducted using the  cut-off $a_\alpha$  in (\ref{eq:alpha})-(\ref{aa}). Then 
	\begin{equation*}
	\beta=\Pr(Z_i > a_\alpha | 1 - \tau)\leq \Pr(Z_i \leq a_\alpha | \tau)\leq \alpha
	\end{equation*} 
	provided that  $\tau > \tau^*$ for some $\tau^*>0.5$ in which case,  $a_\alpha \ge \frac {N_1-1}2$. 
\end{restatable}

Lemma \ref{lem:aControlsb} states that when $a_\alpha \ge \frac{N_1-1}{2}$, the type I and type II error can be controlled simultaneously through $\alpha$. 
It can be similarly  shown that an analogous properties also exist when explicitly controlling the type II error rate.   
Accordingly, the algorithm thus has optimal behavior when $\tau$ is sufficiently far from 0.5 to ensure that $a_\alpha \ge \frac{N_1-1}{2}$.
For a given value of $N_1$ and $\alpha$,  the  bound $\tau^*$ on $\tau$ can be explicitly  calculated using polynomial solvers or the normal approximation to the binomial, depending on the magnitude of $N_1$.

\subsection{Estimation} \label{sec:estimation}

We note that both $\bar F$ and $\bar G$ are generally unknown (as are $t^*$ and $\psi$), but can easily be estimated non-parametrically from the available data by their respective empirical c.d.fs, for sufficiently large $N_1$ and $N-N_1$.
For each given instance $i\in C_1$, 
\begin{align}
\hat G_i(t) & := \frac1{ N_{1}-1}\sum_{j\in C_1, \ j\neq i}^{ N_{1}} \indicator[d_{ij}\leq t],&
\hat F_i(t) & :=\sum_{k = 2}^K   \frac{ N_{k}}{N -  N_1}\cdot\hat F_i^{(k)}(t) \label{eq:GF_estimators}
\end{align}
where, for each $k=2, 3, \dots, K$, 
\[
\hat F_i^{(k)}(t)  := \frac1{ N_{k}}\sum_{j \in C_{k}}\indicator[d_{ij}\leq t]. 
\]
Clearly, $\hat G_i(t)$ and $\hat F_i(t)$ are empirical c.d.fs for estimating, based on the $i^{th}$ instance, $\bar G(t)$ and $\bar F(t)$ respectively.
Accordingly, when combined, 
\begin{align} \label{eq:GFbar_estimators}
\hat{\bar F}(t) &= \frac1{ N_1}\sum_{i=1}^{ N_{1}} \hat F_{i}(t), & \hat{\bar G}(t)&= \frac1{N_{1}}\sum_{i=1}^{ N_{1}} \hat G_{i}(t),
\end{align}
are the estimators of $\bar F(t)$ and $\bar G(t)$, respectively. 
Further, in similarity to (\ref{eq:psi}), we set 
\begin{equation}\label{eq:psi_hat}
\hat \psi(t) := \hat{\bar G}^{^{-1}}(1-\hat{\bar F}(t)), 
\end{equation}
and we let $\hat t^*$ denote the ``solution'' of $\hat \psi (t^*_c)=t^*_c$; that is 
\begin{equation}\label{eq:t_estimator}
\hat t^*:=\inf_{t}\{\hat \psi (t)\leq t\}.
\end{equation}
Clearly, the value of $\tau$ in (\ref{eq:tau}) would be estimated by 
\begin{equation} \label{eq:tau_hat}
\hat {\tau}=\hat{\bar G}(\hat t^*).
\end{equation} 
Note that in view of (\ref{eq:GF_estimators}), $N_1 \hat \tau\equiv \sum_{i=1}^{N_1} \hat \tau_i$, with  $\hat \tau_i\equiv \hat G_i(\hat t^*)$ for $i=1, 2, \dots, N_1$.
With $\hat t^*$ as an estimate of $t^*$ in (\ref{eq:tau}), we have that $Z_i\equiv (N_1-1)\hat G_i(\hat t^*)$ and $\hat \tau_i\equiv Z_i/(N_1-1)$ and therefore an equivalent estimate of $\hat \tau$ is
\begin{equation}\label{eq:tau_hat2}
\hat \tau= \frac{1}{N_1}\sum_{i=1}^{N_1}\frac{Z_i}{N_1-1}.
\end{equation}
By Lemma \ref{lem:nulldist} (a), $\E\left[\hat\tau_i | \mathcal H_0^{(i)}\right] =(N_1-1) \tau$ and hence, $\hat \tau$ in (\ref{eq:tau_hat2}) is an unbaised estimator of $\tau$,  $	\E\left[\hat\tau | \mathcal H_0^*\right] = \tau$.

\section{ A Simulation Study} \label{sec:simulation}
Our simulation study is designed to mimic the conditions of a LC-MS/MS shotgun proteomics study.
In this light, we consider a set-up in which $N_1$ instances (peptides) belong to class/protein $C_1$, and $N_2$ instances model the peptides belonging to any other class/protein.
The distance between instances is measured using correlation distance, again mimicking a common way to measure similarity between peptides, although  similar results can be obtained using other (quasi-) distance metrics (e.g. Euclidean distance).

\subsection{The Simulation setup}

To establish some notation, suppose that $\y=(y_1, y_2, \dots, y_N)^\prime$  is a $N\times 1$ random vector having some joint distribution $H_N$. We assume, without loss of generality, that the values of $\y$ are standardized, so that $E(y_i)=0$ and $V(y_i)=1$, for each $i=1, 2, \dots, N$. We  denote by $\cD$ the corresponding correlation (covariance) matrix for $\y$, $\cD=\cor(\y, \y^\prime)$. To simplify, we assume that $K=2$ so that the $N$ instances are presumed to belong to either class $C_1$ or $C_2$. Accordingly, we partitioned $\y$  and $\cD$ as $\y=[\y_{1}^\prime, \y_2^\prime]^\prime$, with, $\y_1=(y_{1,1}, y_{1,2}, \dots, y_{1,N_{1}})^\prime$ and $\y_2=(y_{2,N_1+1}, y_{2,N_1+2}, \dots, y_{2,N_1+N_{2}})^\prime$, $N_1+N_2=N$, and 
\[
\cD= \begin{bmatrix}
\cD_{1,1} & \cD_{1,2} \\ 
\cD_{2,1} & \cD_{2,2} 
\end{bmatrix},
\]
with $\cD_{k,\ell}=\cor(\y_k, \y_\ell^\prime), \, k,\ell=1,2$.

As in Section \ref{sec:methodology},  let $\X$ denote the $n\times N$ data matrix of the observed intensities. We denote by $\wu{\xe}^{_\prime}_j:=(x_{j1}, x_{j2}, \dots, x_{jN})$ the $j^{th}$ row of $\X$, $j=1, 2, \dots, n$, and 
we assume that $\wu{\xe}^{_\prime}_1, \wu{\xe}^{_\prime}_2, \dots, \wu{\xe}^{_\prime}_n$are independent and identically distributed as $\y\sim H_N$. 

Using standard notation, we write, $\mathbf{1}_n=(1, 1, \dots, 1)^\prime$,  $\I_n$  for the $n\times n$ identity matrix and $\J_n={\mathbf{1}_n \mathbf{1}_n^\prime}$ for the $n\times n$ matrix of $1$s.
For the simulation studies we conducted, we took $H_N$ to be the $N$-variate normal distribution, so that  
\[
\wu{\xe}^{_\prime}_j\sim \N_{N} ({\mathbf{0}}, \, \cD), \ j=1, 2, \dots, n, \ \ \ i.i.d.,
\]
where 
\begin{align}
\cD_{1,1} &=(1-\rho_1)\I_{N_1}+\rho_1\J_{N_1} \label{eq:D11} \\ 
\cD_{2,2} &=(1-\rho_2)\I_{N_2}+\rho_2\J_{N_2} \label{eq:D22} \\ 
\cD_{2,1 }&=\cD^\prime_{1,2}=\rho_{12}{\mathbf{1}_{N_2} \mathbf{1^\prime}_{N_1}} 
\end{align}
for $\boldsymbol \rho = (\rho_1, \rho_{12}, \rho_2)$ with $0\leq \rho_{12}\le\rho_2\le\rho_1<1$.

To allow for misclassification of instances, we included for a certain proportion $p$, some $m:=[pN_1]$ of the $N_1$ `observed' instances intensities from $C_1$ that were actually simulated with $\cD_{1,1}$ being replaced by $\cD^*_{1,1}=(1-\rho_2)\I_{N_1}+\rho_2\J_{N_1}$  in (\ref{eq:D11}) above.
Thus, $m$ is the number of misclassified instances among the $N_1$ instances that were initially labeled as belonging to $C_1$. 

\begin{remark} \label{rem:rho-proxy}
	In this simulation, $\rho_2$ reflects (as a proxy) the common characteristics of two mislabeled instances.
	When $\rho_2 = \rho_1$, distances between two mislabeled instances have the same distribution as two correctly labeled instances,   as would be the case for binary classification when $\rho_2 = \rho_1$.
	When $\rho_2 = \rho_{12}$,  distances for two mislabeled instances have the same distribution as a distances between a correctly labeled instance and a mislabeled instance.  
	This would be the case if the probability that two mislabeled instances come from the same class is zero. 
\end{remark}

For each simulation run, we recorded $\hat \tau$, $\hat t^*$, and counted the number of true positives (TPs), true negatives (TNs), false positives (FPs), and false negatives (FNs), as defined in Table \ref{tab:contigency-table}.
From this data, we calculated the sensitivity, specificity, false discovery rate (FDR), false omission rate (FOR), and percent reduction in FOR (\pc) for each run; defined as follows:
\begin{table}
	\centering
	\caption{\label{tab:contigency-table} For a single run, each instance has one of four possible outcomes.  The notation for the total count of instances with each these outcomes in a single run.}
 	\begin{tabular}{cc|C{1in}C{1in}|c}
		& & \multicolumn{2}{c|}{Truth} & \\
		& & Correctly labeled & mislabeled & Total \\\hline 
		\multirow{2}{*}{\rotatebox{90}{Result}}
		& Keep \rule{0pt}{12pt}  & TN & FN & $N_1 - R$ \\[.7ex]
		& Remove & FP & TP & $R$ \\[.7ex]\hline
		& & $N_1 - m$ & $m$ & $N_1$\rule{0pt}{12pt}
	\end{tabular}
\end{table}

\begin{center}
	\begin{tabular}{r @{\hspace{.2cm}=\hspace{.2cm}} l p{3in} }
		Sensitivity & $\displaystyle \frac{\tp}{\tp + \fn}$ & Proportion of correctly removed instances out of all mislabeled instances. \\
		Specificity &  $\displaystyle \frac{\tn}{\tn + \fp}$ & Proportion of correctly retained instances out of all correctly labeled instances.\\
		FDR& $\displaystyle \frac{\fp}{\max(\tp + \fp, 1)}$ & Proportion of correctly removed instances out of all those removed.\\
		FOR & $\displaystyle \frac{\fn}{\max(\tn + \fn, 1)}$ & Proportion of correctly retained instances out of all those retained. \\
		\pc & $\displaystyle \left(1 - \frac{\for}{p}\right) \times 100 $ & Percent reduction in FOR relative to $p$.
	\end{tabular}
\end{center}
Each statistic was averaged over all 1000 runs.

\subsection{Simulation results}

 We conducted $B=1000$ simulation runs of the test procedure with $n=$10, 25, 50, 75, 100, 250, 500, 700; $N_1=$25, 50, 100, 500; $N_2=1000$; $\alpha_0=0.05$; and $\boldsymbol\rho =(\rho_1 = 0.5, \rho_{12} = 0.1, \rho_2 = 0.5)$, $(0.5, 0.1, 0.1)$, $(0.5, 0.2, 0.5)$, and $(0.5, 0.2, 0.2)$.
We also varied the value of $p$, the proportion of mislabeled instances such that  $p=0.0,\,0.05,\,0.1,\,0.2,\,0.25$.
In particular, $p = 0$ means no mislabeling and that the initial labeling is perfect.

\subsubsection{Results with no mislabeling (case $p=0$)} \label{sec:sim-no-mislabeling}
Simulations where $p = 0$ (i.e. no mislabeling) were used to illustrate the theoretical assumptions of the testing procedure, as this presents a case in which the global null hypothesis in (\ref{eq:bigH}) holds.
In this case in particular, the FDR measures the proportion of incorrectly rejected hypotheses in (\ref{H0i}) out of all rejected hypothesis tests, given that at least one hypothesis test was rejected.
\begin{remark}\label{rem:fdr-p0}
For $p = 0$, \emph{every} rejected hypothesis test in (\ref{H0i}) is incorrectly rejected.
Thus, the FDR is 1 if \emph{any} rejected hypothesis tests are observed, and zero otherwise (by the definition of the FDR).
Consequently, the average value of the FDR over all $B$ runs provides an estimate of $\alpha^\prime$.
\end{remark}

\begin{figure}
	\centering
	\subfloat[$\rho_{12}=0.1$]{\includegraphics[width = .49\textwidth]{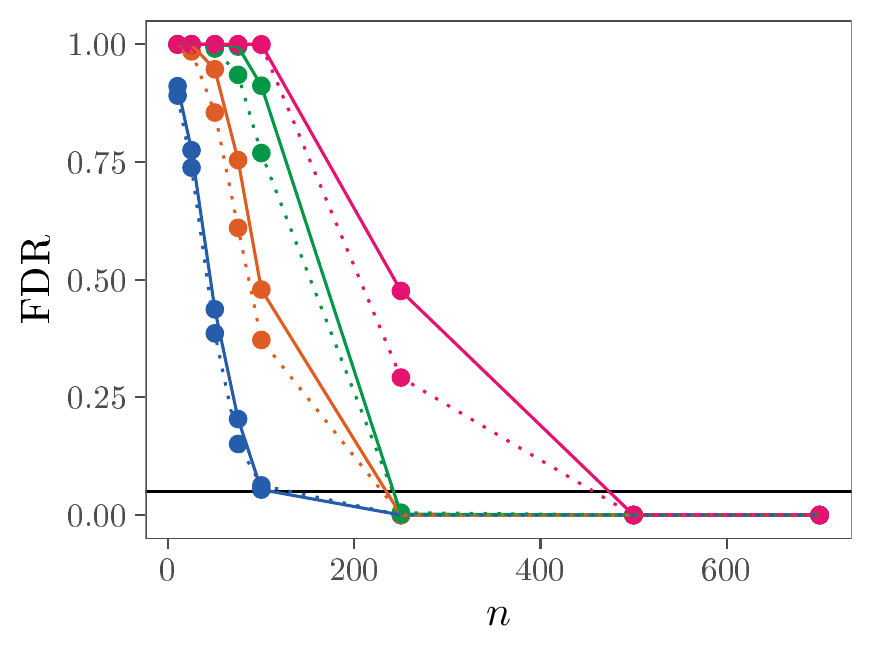} \label{subfig:rho12-1-p0-sim-fdr}}
	\subfloat[$\rho_{12}=0.2$]{\includegraphics[width = .49\textwidth]{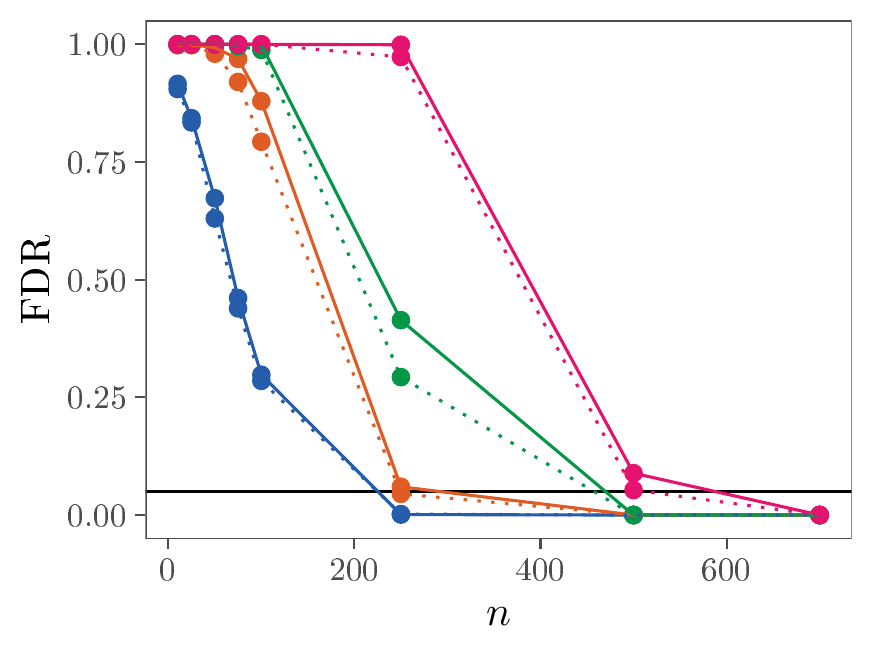} \label{subfig:rho12-2-p0-sim-fdr}}
	
	\includegraphics[width = \textwidth]{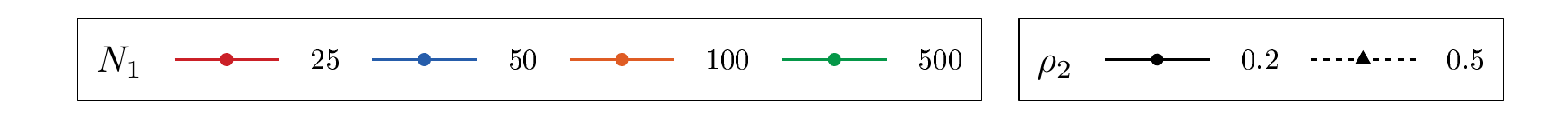}
	\caption{
		\label{fig:sim_p=0_fdr} The $\fdr$ as a function of $n$ and $N_1$ for \protect\subref{subfig:rho12-1-p0-sim-fdr} $\rho_{12} = 0.1$ and \protect\subref{subfig:rho12-2-p0-sim-fdr} $\rho_{12} = 0.2$ where $p = 0$. 
	}
\end{figure}

Figure \ref{fig:sim_p=0_fdr} and Table \ref{tab:sim1_p0} show the FDR for various $n$,  $N_1$, and $\rho_{12}$.  
As seen in the figures, the FDR converges to zero as $n$ increases for all values of $N_1$, but the convergence slows as $N_1$ increases. 
When $n$ is small, at least one instance was removed from almost all classes.
We attribute this behavior to the inherent correlation structure of the data.
This will be explored further in Section \ref{sec:sim-indepencence}.

\begin{table}
	\centering
	\begin{threeparttable}
		\caption{Simulation results for $p = 0$, $\rho_{12} = \rho_2 = 0.2$, and $\alpha_0 = 0.05$. 
			The sensitivity and FOR are excluded from the table because they are either undefined or constant when $p = 0$.\label{tab:sim1_p0}}
		\begin{tabular}{rcllcllcllcll}
			\toprule
			\multicolumn{1}{c}{}&\multicolumn{1}{c}{}&\multicolumn{2}{c}{$N_1 = 25$}&\multicolumn{1}{c}{}&\multicolumn{2}{c}{$N_1 = 50$}&\multicolumn{1}{c}{}&\multicolumn{2}{c}{$N_1 = 100$}&\multicolumn{1}{c}{}&\multicolumn{2}{c}{$N_1 = 500$}\\
			\cline{3-4} \cline{6-7} \cline{9-10} \cline{12-13}
			\multicolumn{1}{c}{$n$}&\multicolumn{1}{c}{}&\multicolumn{1}{c}{FDR\tnote{a}}&\multicolumn{1}{c}{Spec.}&\multicolumn{1}{c}{}&\multicolumn{1}{c}{FDR\tnote{a}}&\multicolumn{1}{c}{Spec.}&\multicolumn{1}{c}{}&\multicolumn{1}{c}{FDR\tnote{a}}&\multicolumn{1}{c}{Spec.}&\multicolumn{1}{c}{}&\multicolumn{1}{c}{FDR\tnote{a}}&\multicolumn{1}{c}{Spec.}\\
			\midrule
			10 && 0.916 & 0.929 && 1.000 & 0.880 && 1.000 & 0.834 && 1.000 & 0.737\\
			25 && 0.843 & 0.947 && 0.999 & 0.904 && 1.000 & 0.860 && 1.000 & 0.766\\
			50 && 0.673 & 0.965 && 0.993 & 0.933 && 1.000 & 0.894 && 1.000 & 0.807\\
			75 && 0.461 & 0.979 && 0.969 & 0.952 && 1.000 & 0.921 && 1.000 & 0.841\\
			100 && 0.298 & 0.987 && 0.879 & 0.967 && 0.999 & 0.943 && 1.000 & 0.872\\
			250 && 0.001 & 1.000 && 0.060 & 0.999 && 0.414 & 0.995 && 0.999 & 0.977\\
			500 && 0.000 & 1.000 && 0.000 & 1.000 && 0.000 & 1.000 && 0.089 & 1.000\\
			700 && 0.000 & 1.000 && 0.000 & 1.000 && 0.000 & 1.000 && 0.000 & 1.000\\
			\bottomrule
		\end{tabular}
		\begin{tablenotes}
			\item[a] See Remark 2
		\end{tablenotes}
	\end{threeparttable}
\end{table}

The specificity measures the ability to keep instances which are correctly labeled.
From Figure \ref{fig:sim_p=0_specificity} and Table \ref{tab:sim1_p0}, it can be seen that while almost all runs remove at least one instance (based on the FDR), most instances are retained.
With $n = 10$, $\rho_{12} = \rho_2 = 0.2$, and $N_1 = 25$, an average of 23.2 out of 25 were retained in each run.
This number decreased as $N_1$ increased, corresponding to the slower convergence of the FDR when $N_1$ is larger.
Even in this case, for $n = 10$, $\rho_{12} = \rho_2 = 0.2$ and $N_1 = 500$, an average of 368.5 instances are retained each time.

For $p = 0$, the sensitivity is undefined and the FOR is universally zero.
These statistics are relevant only when at least one mislabeled instance is present in the data set (i.e. $p > 0$) and therefore are omitted from Table \ref{tab:sim1_p0}.

\begin{figure}
	\centering
	\subfloat[$\rho_{12}=0.1$]{\includegraphics[width = .49\textwidth]{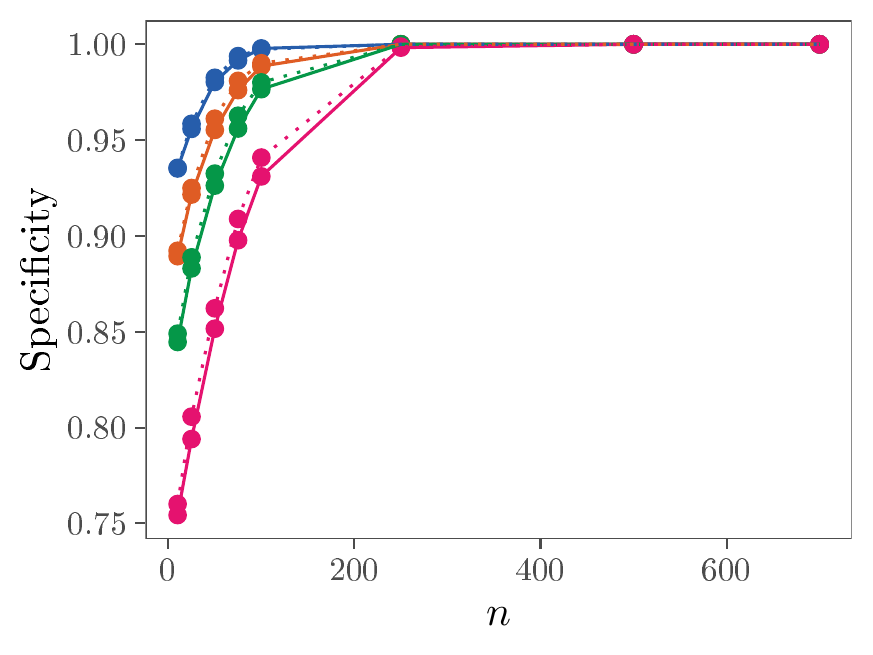} \label{subfig:rho12-1-p0-sim-specificity}}
	\subfloat[$\rho_{12}=0.2$]{\includegraphics[width = .49\textwidth]{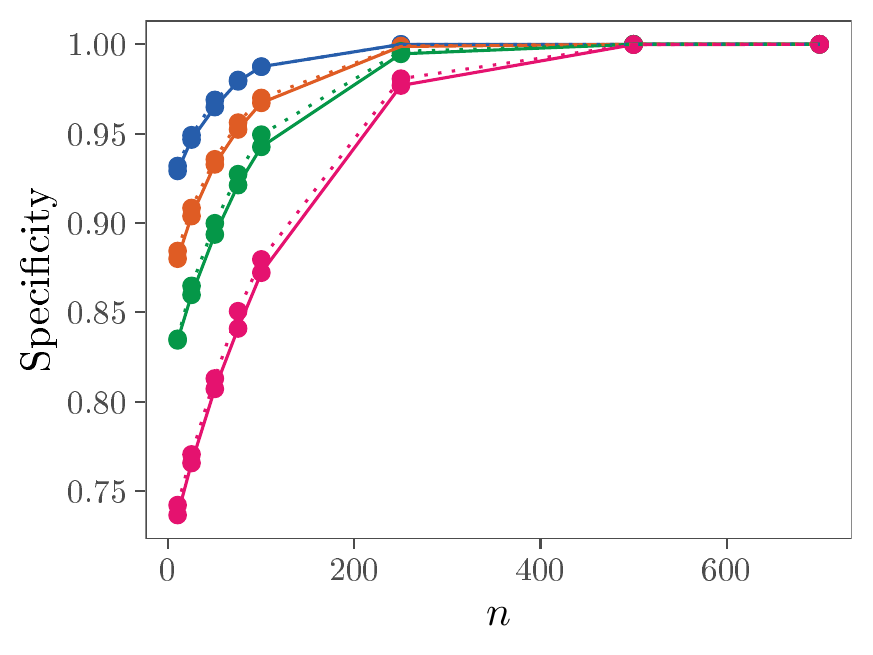} \label{subfig:rho12-2-p0-sim-specificity}}
	
	\includegraphics[width = \textwidth]{pg0_legend-only}
	\caption{
		\label{fig:sim_p=0_specificity} The specificity as a function of $n$, $N_1$ for \protect\subref{subfig:rho12-1-p0-sim-specificity} $\rho_{12} = 0.1$ and \protect\subref{subfig:rho12-2-p0-sim-specificity} $\rho_{12} = 0.2$ where $p = m = 0$.    
	}
\end{figure}

\begin{figure}
	\centering
	\includegraphics[width = .6\textwidth]{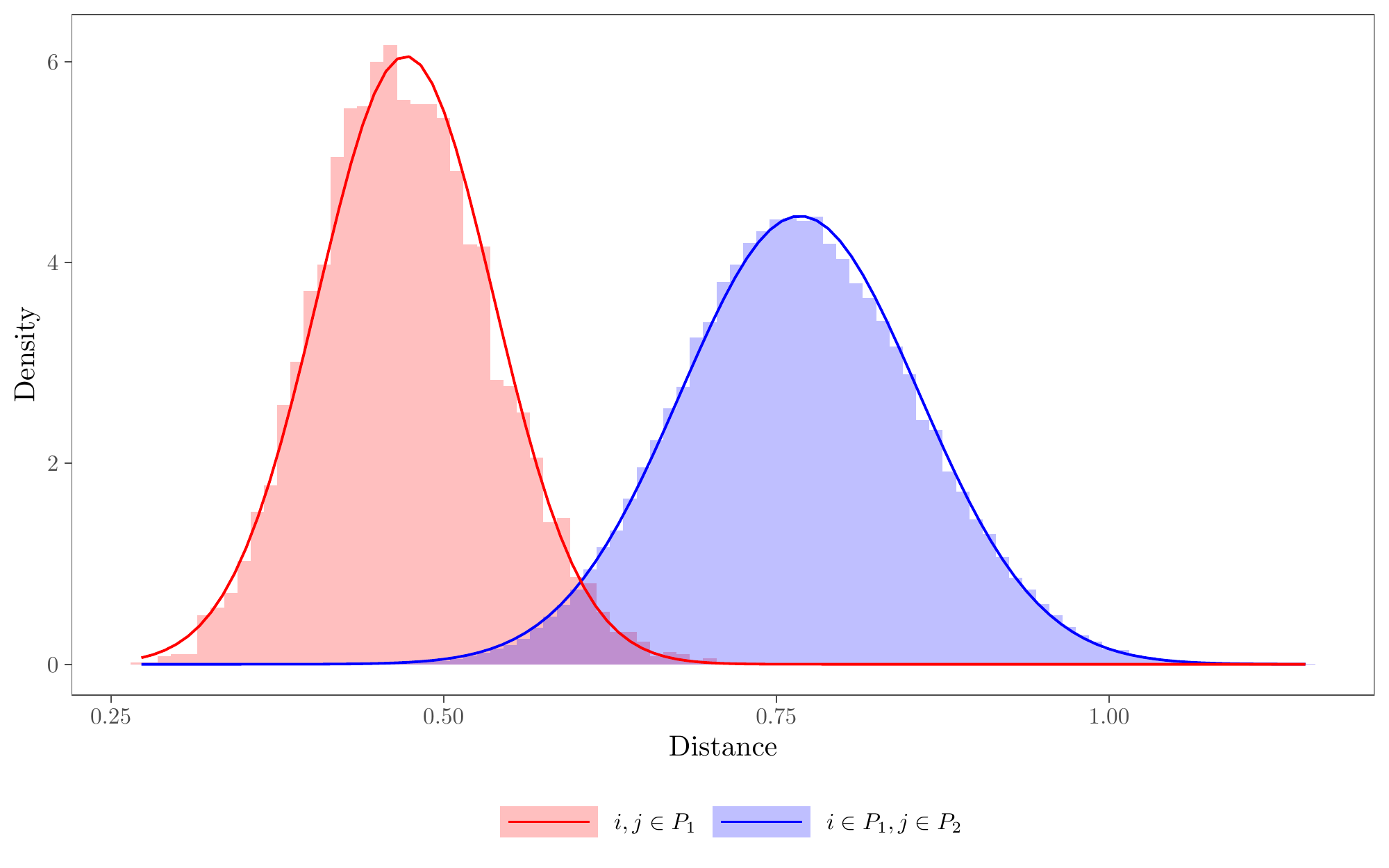}
	\caption{\label{fig:simulated-independent-histogram} A histogram of all distances within $C_1$ (red) and between instances in $C_1$ and those in $C_2$ (blue) for $p = 0$, $N_1 = 500$, $n = 100$, and $\rho_{12} = \rho_2 = 0.2$.  The lines are generated from a normal distribution with mean and standard deviations matched to the data.}
\end{figure}

\subsubsection{Behavior under artificially constructed independence} \label{sec:sim-indepencence}
In light of Remark \ref{rem:fdr-p0} and the likely  impact of the correlation structure present in the data on the FDR, we designed a simulation study to explore this effect.
In this study the ``distance'' matrices were artificially created in a manner which preserved dependence due to symmetry but removed all other dependencies across distances.

To simulate ``distance'' matrices in this case, we began by randomly generating a distance matrix using the original test procedure with $p = 0$ and $n = 100$, $N_1 = 500$, $N_2 = 1000$, and $\boldsymbol{\rho} = (0.5, 0.2, 0.2)$.
A normal distribution was fit to the within-$C_1$ distances and the $C_1$ to $C_2$ distances, as shown in Figure \ref{fig:simulated-independent-histogram}.
For $N_1 =$ 25, 50, 100, 500, 1000, 2000, 5000, and 12000 and $N_2 = 1000$,  these normal distributions were used to generate $B$ new distance matrices by drawing $d_{ij}$ for $1 \le i < j \le N$ as follows:
\[
d_{ij} \sim \begin{cases}
0                                & i = j \\
N(\mu = 0.523,\ \sigma = 0.0684) & 1 \le i < j \le N_1 \\
N(\mu = 0.771,\ \sigma = 0.0903)  & 1 \le i \le N_1 < j \le N \\
d_{ji}                            & i > j
\end{cases}.
\]

Figure \ref{fig:independent-simulation-convergence} shows the results of this procedure on five sets of $B = 1000$ runs.
As $N_1$ increased, the FDR also increased, but never surpassed the theoretical limit of $1 - e^{-0.05}$, consistent with the theoretical result given in (\ref{eq:ind-alpha}).

\begin{figure}
	\centering
	\includegraphics[width = .7\textwidth]{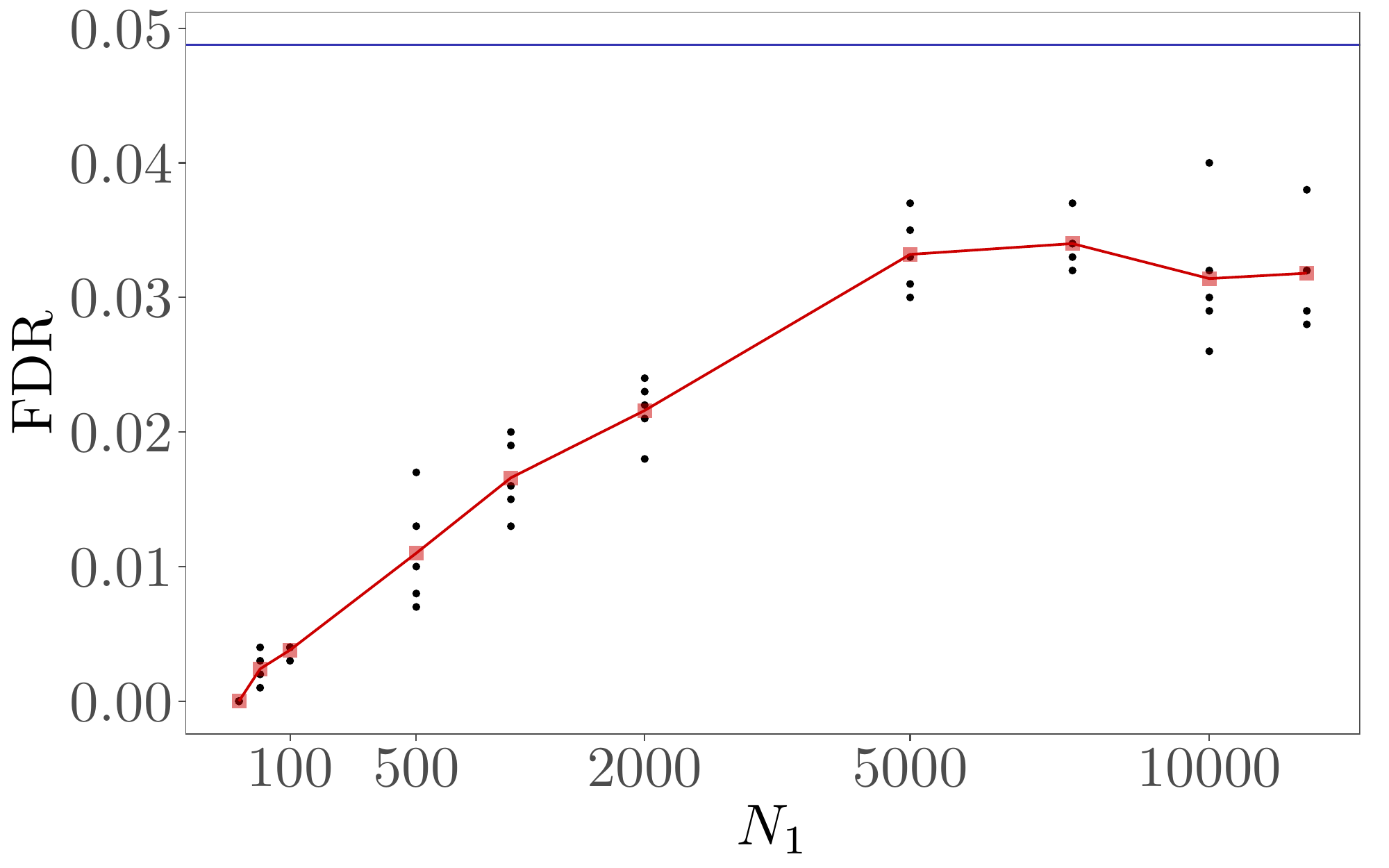}
	\caption{\label{fig:independent-simulation-convergence} Direct simulation of the distance matrix using independent normal draws from a normal distribution and $p = 0$, performed as five batches of $B = 1000$ runs each (shown in black).
		The average across all five batches is shown in red.
		The blue line at the top of the plot shows $1 - e^{-\alpha_0}$ for $\alpha_0 = 0.05$.
	}
\end{figure}

\subsubsection{Results under some mislabeling, (case  $p > 0$)}
To develop context regarding the results with some mislabeling (i.e. $p > 0$), we first consider for a moment a trivial filtering procedure which retains all $N_1$ instances in $C_1$.
Since the $pN_1$ incorrect instances are retained, the FOR of the trivial procedure is $pN_1 / N_1 = p$.
Clearly, \emph{an FOR of $p$ can be achieved without performing any filtering at all simply by returning all $N_1$ instances in $C_1$}.
For a testing procedure to improve upon this nominal level, it must result in an FOR below $p$.
With  $p > 0$, the FDR calculates the proportion of correctly labeled instances among those removed.
However, in the context of proteomics, where the set of retained peptides are used in subsequent analyses, we found that the specificity was a more relevant metric.
Consequently, the FOR and specificity are the primary statistics used to evaluate our proposed testing procedure in the presence of labeling errors ($p > 0$), with \pc providing a standardized method of evaluating the decrease in FOR in a manner independent of $p$.

\begin{figure}
	\centering
	\subfloat[$p=0.05$]{\includegraphics[width = .49\textwidth]{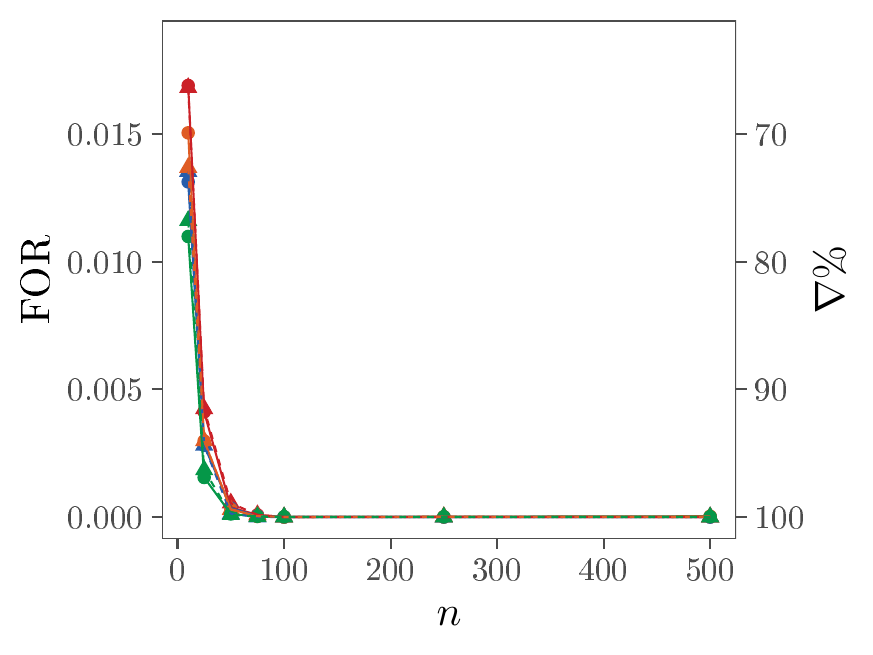} \label{subfig:for_p5}}
	\subfloat[$p=0.10$]{\includegraphics[width = .49\textwidth]{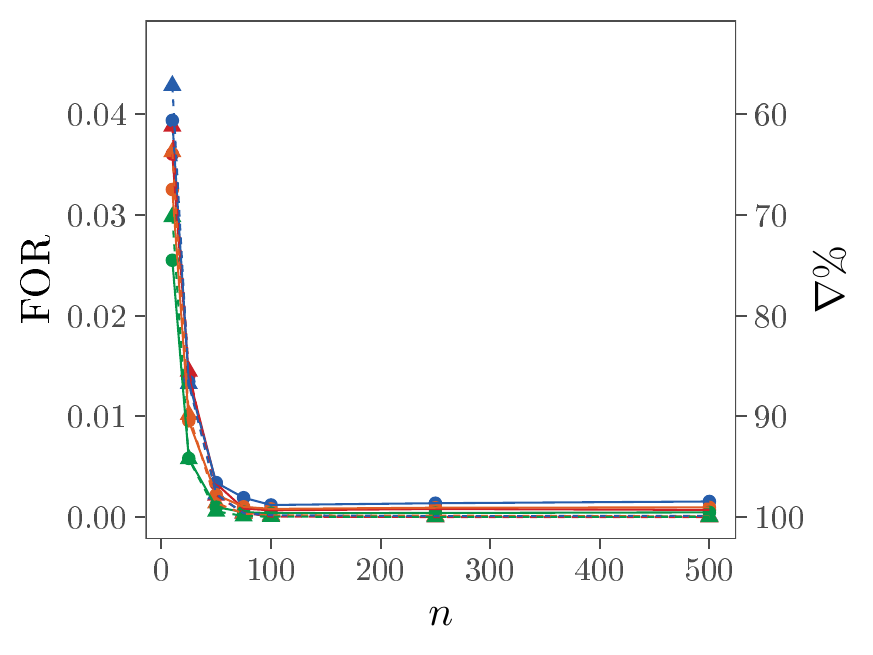} \label{subfig:for_p10}}\\
	\subfloat[$p=0.20$]{\includegraphics[width = .49\textwidth]{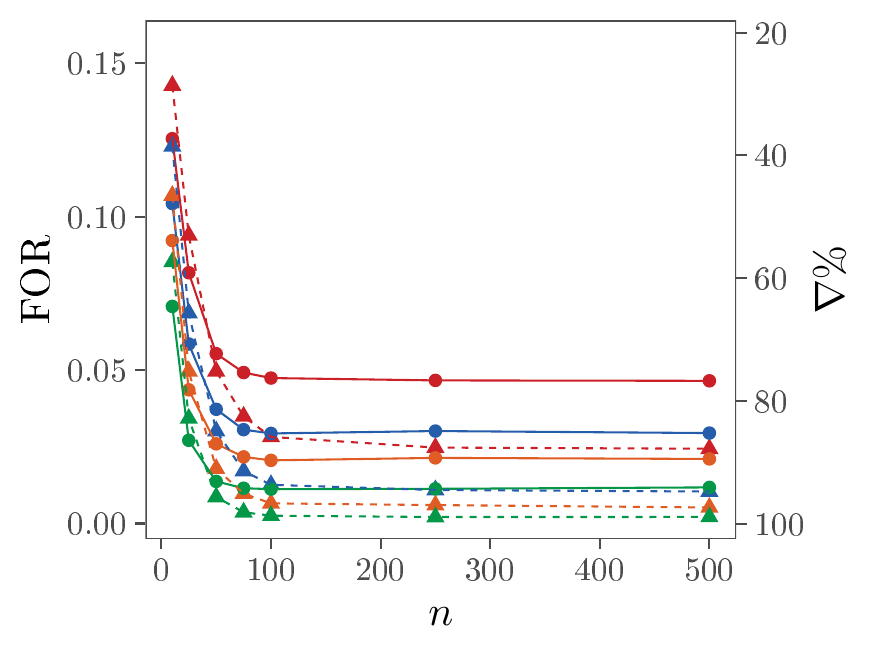} \label{subfig:for_p20}}
	\subfloat[$p=0.25$]{\includegraphics[width = .49\textwidth]{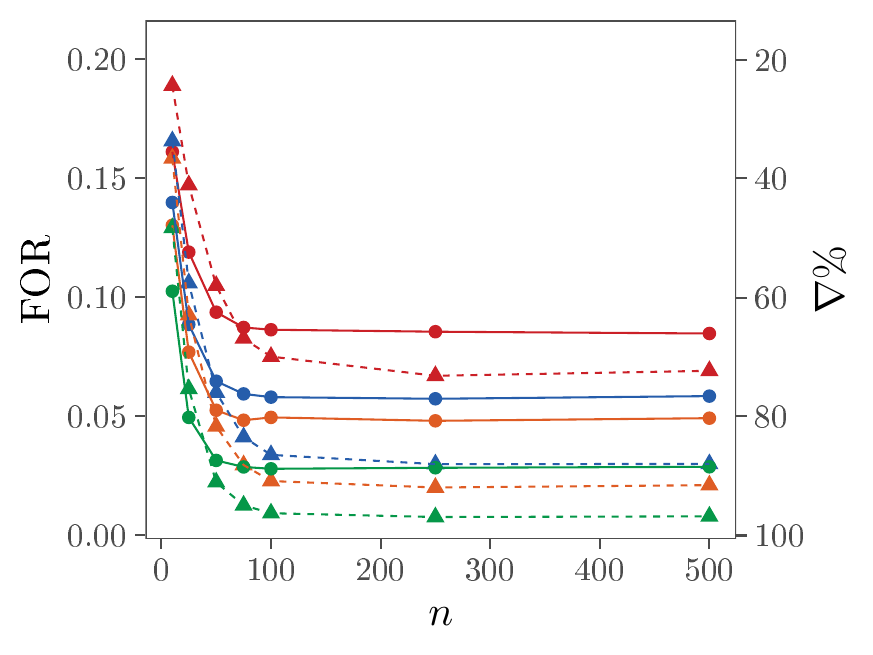} \label{subfig:for_p25}}
	
	\includegraphics[width = \textwidth]{pg0_legend-only}
	\label{subfig:rho12=.2-rho2=.5-p0-sim_fdr}

	\caption{\label{fig:sim_p>0_for} The FOR as a function of $n$ and $N_1$ for $p = 0.05, 0.10, 0.20,$ and $0.25$ with $\rho_{12}= 0.2$.}
\end{figure}

Figure \ref{fig:sim_p>0_for} provides the average FOR and \pc over the $B =1000$ simulation runs as a function of $n$ for $\rho_{12} = 0.2$ and varying combinations of $N_1$, $\rho_2$ and $p$.
In all cases, the procedure reduced the FOR relative to $p$ (that is, $\pc > 0$ for all results).
Small values of $n$ produced the smallest reduction (highest FOR), but this converged at approximately $n = 100$ to a value dependent on $p$, $N_1$, and $\rho_2$.
Table \ref{tab:sim_convergence} shows how the average FOR compares between small sample sizes ($n = 10$) and large sample sizes (averaged across $n = 250$, 500, and 700) at each value of $p$, $N_1$, and $\rho_2$.
For $p = 0.05$, the FOR converged to 0 for all values of $N_1$ and $\rho_2$.
For higher values of $p$, decreasing $N_1$ and decreasing $\rho_2$ caused the FOR to be higher for $n \ge 250$.
For example, when $p = 0.25$ and $N_1 = 500$, for example, the instances remaining in the class after filtering will still include \SI{10.26}{\percent} (\pc = 59.0) mislabeled instances for $n = 10$ and \SI{2.85}{\percent} (\pc = 88.6) mislabeled instances when $n$ is large.  
These are both substantial decreases from \SI{25}{\percent} mislabeled instances as seen in the unfiltered data.
\begin{table}
	\centering
	\caption{\label{tab:sim_convergence} The mean $\for$ and $\pc$ at $n = 10$ and $n \ge 250$ for each combination of $p$, $N_1$, and $\rho_{2}$ at $\rho_{12}= 0.2$.
	}
	\begin{tabular}{rrrrrrrrrr}
		\toprule
		\multicolumn{1}{c}{} & \multicolumn{1}{c}{} & \multicolumn{4}{c}{$\rho_{2} = 0.2$} & \multicolumn{4}{c}{$\rho_{2} = 0.5$} \\
		\cmidrule(l{3pt}r{3pt}){3-6} \cmidrule(l{3pt}r{3pt}){7-10}
		\multicolumn{1}{c}{} & \multicolumn{1}{c}{} & \multicolumn{2}{c}{$n = 10$} & \multicolumn{2}{c}{$n \ge 250$} & \multicolumn{2}{c}{$n=10$} & \multicolumn{2}{c}{$n \ge 250$} \\
		\cmidrule(l{3pt}r{3pt}){3-4} \cmidrule(l{3pt}r{3pt}){5-6} \cmidrule(l{3pt}r{3pt}){7-8} \cmidrule(l{3pt}r{3pt}){9-10}
		$p$ & $N_1$ & FOR & $\%\Delta$ & FOR & $\%\Delta$ & FOR & $\%\Delta$ & FOR & $\%\Delta$\\
		\midrule
		& 25 & 0.0169 & 66.2 & 0.0000 & 100.0 & 0.0168 & 66.4 & 0.0000 & 100.0\\
		
		& 50 & 0.0131 & 73.8 & 0.0000 & 100.0 & 0.0135 & 72.9 & 0.0000 & 100.0\\
		
		& 100 & 0.0150 & 69.9 & 0.0000 & 100.0 & 0.0137 & 72.6 & 0.0000 & 100.0\\
		
		\multirow[t]{-4}{*}{\raggedleft\arraybackslash 0.05} & 500 & 0.0110 & 78.0 & 0.0000 & 100.0 & 0.0116 & 76.8 & 0.0000 & 100.0\\
		\cmidrule{1-10}
		& 25 & 0.0361 & 63.9 & 0.0007 & 99.3 & 0.0388 & 61.2 & 0.0000 & 100.0\\
		
		& 50 & 0.0394 & 60.6 & 0.0015 & 98.5 & 0.0429 & 57.1 & 0.0001 & 99.9\\
		
		& 100 & 0.0325 & 67.5 & 0.0009 & 99.1 & 0.0363 & 63.7 & 0.0001 & 99.9\\
		
		\multirow[t]{-4}{*}{\raggedleft\arraybackslash 0.10} & 500 & 0.0255 & 74.5 & 0.0004 & 99.6 & 0.0298 & 70.2 & 0.0000 & 100.0\\
		\cmidrule{1-10}
		& 25 & 0.0585 & 61.0 & 0.0047 & 96.8 & 0.0635 & 57.7 & 0.0006 & 99.6\\
		
		& 50 & 0.0637 & 57.5 & 0.0068 & 95.5 & 0.0694 & 53.8 & 0.0011 & 99.3\\
		
		& 100 & 0.0585 & 61.0 & 0.0061 & 95.9 & 0.0659 & 56.1 & 0.0010 & 99.3\\
		
		\multirow[t]{-4}{*}{\raggedleft\arraybackslash 0.15} & 500 & 0.0446 & 70.3 & 0.0032 & 97.8 & 0.0511 & 66.0 & 0.0003 & 99.8\\
		\cmidrule{1-10}
		& 25 & 0.1254 & 37.3 & 0.0468 & 76.6 & 0.1427 & 28.6 & 0.0249 & 87.6\\
		
		& 50 & 0.1042 & 47.9 & 0.0300 & 85.0 & 0.1229 & 38.6 & 0.0108 & 94.6\\
		
		& 100 & 0.0922 & 53.9 & 0.0214 & 89.3 & 0.1069 & 46.6 & 0.0057 & 97.2\\
		
		\multirow[t]{-4}{*}{\raggedleft\arraybackslash 0.20} & 500 & 0.0707 & 64.6 & 0.0115 & 94.2 & 0.0852 & 57.4 & 0.0021 & 98.9\\
		\cmidrule{1-10}
		& 25 & 0.1612 & 35.5 & 0.0854 & 65.8 & 0.1891 & 24.4 & 0.0678 & 72.9\\
		
		& 50 & 0.1399 & 44.0 & 0.0574 & 77.1 & 0.1657 & 33.7 & 0.0304 & 87.8\\
		
		& 100 & 0.1303 & 47.9 & 0.0485 & 80.6 & 0.1584 & 36.6 & 0.0205 & 91.8\\
		
		\multirow[t]{-4}{*}{\raggedleft\arraybackslash 0.25} & 500 & 0.1026 & 59.0 & 0.0285 & 88.6 & 0.1292 & 48.3 & 0.0078 & 96.9\\
		\bottomrule
	\end{tabular}
	
\end{table}

\begin{figure}
	\centering
	\subfloat[$p=0.05$]{\includegraphics[width = .49\textwidth]{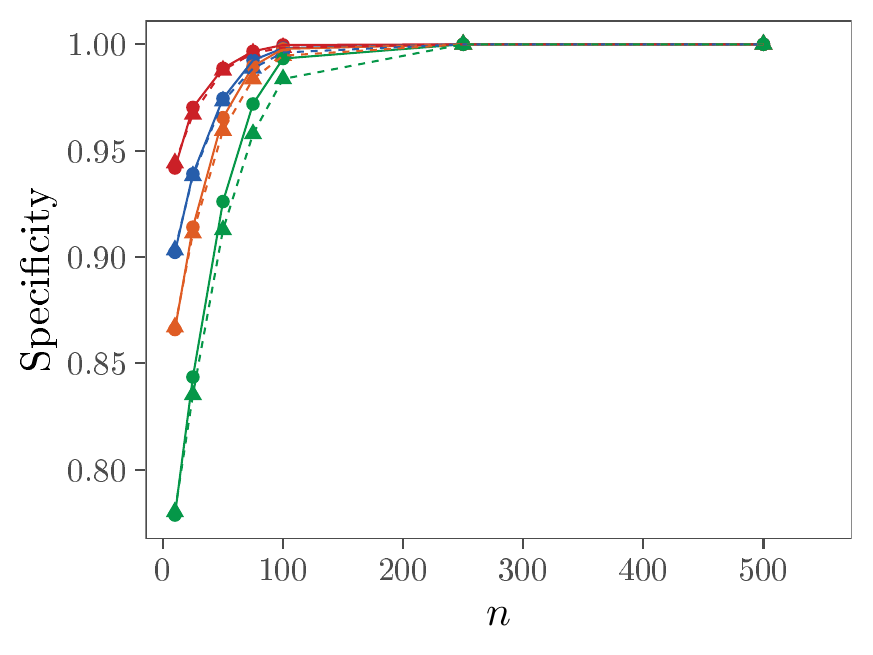} \label{subfig:for-specificity_p5}}
	\subfloat[$p=0.10$]{\includegraphics[width = .49\textwidth]{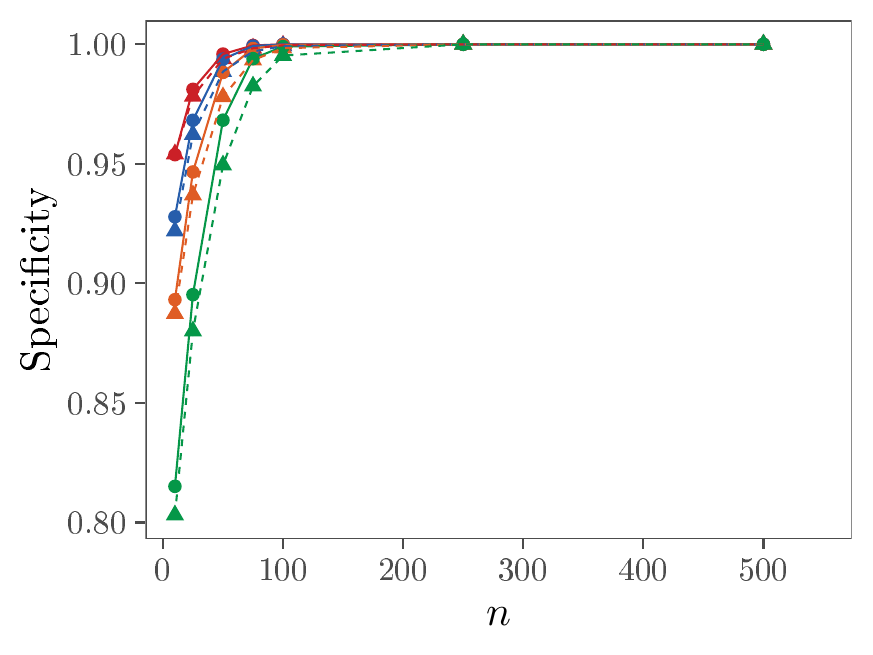} \label{subfig:for-specificity_p10}}\\
	\subfloat[$p=0.15$]{\includegraphics[width = .49\textwidth]{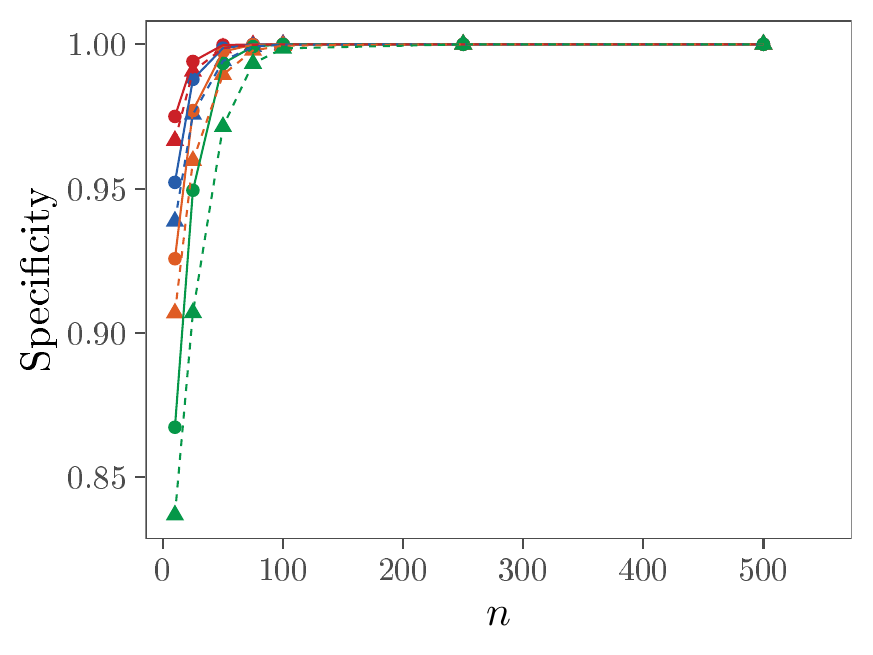} \label{subfig:for-specificity_p20}}
	\subfloat[$p=0.20$]{\includegraphics[width = .49\textwidth]{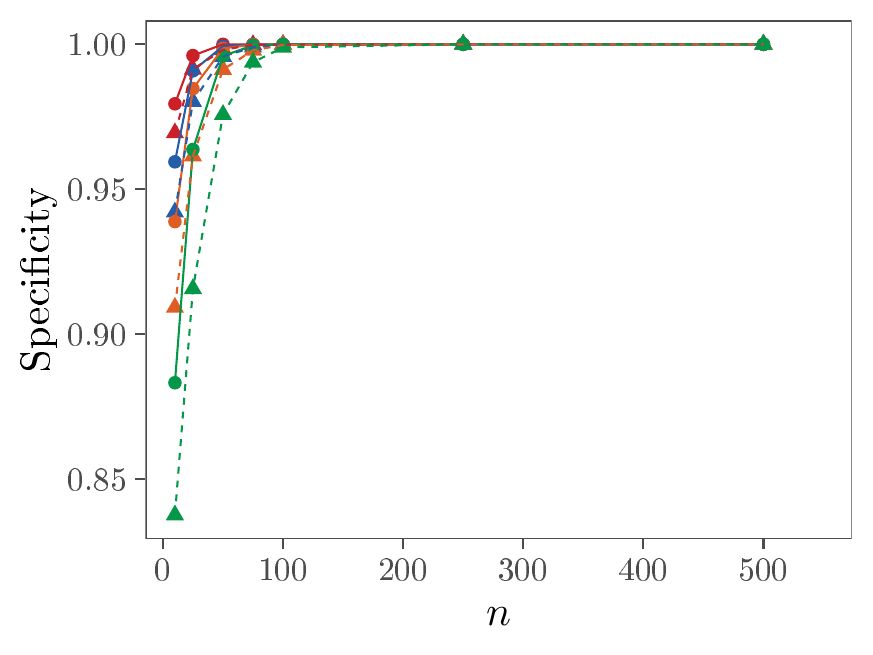} \label{subfig:for-specificity_p25}}
	
	\includegraphics[width = \textwidth]{pg0_legend-only}
	\caption{
		\label{fig:sim_p>0_specificity}
		The specificity as a function of $n$ and $N_1$ for $p = 0.05, 0.10, 0.20,$ and $0.25$ with $\rho_{12}= 0.2$.
	}
\end{figure}

Figure \ref{fig:sim_p>0_specificity} provides the average specificity over the 1000 simulation runs as a function of $n$ for $\rho_{12} = 0.2$ and varying combinations of $N_1$, $\rho_2$ and $p$.
The specificity always converged to 1 as $n$ increased, with convergence by $n = 250$ in all cases.  For larger values of $p$, convergence was faster (by $n = 100$).
For small values of $n$, a higher specificity is observed when $\rho_2$ and $N_1$ is smaller, although even in the worst case, the specificity was greater than 0.75.

Table \ref{tab:sim-pg0-results-n50} gives the average estimate of the FDR, FOR, \pc, sensitivity, and specificity in the case where $n = 50$ and $\rho_{12} = 0.2$ across all combinations of $p$ and $N_1$.
As already noted above, the table shows that the FOR and sensitivity decrease as $p$ and $N_1$ increase.
The FDR gives the proportion of correctly labeled instances among all the removed instances.
This measure increases as a function of $N_1$, corresponding to the decrease in the specificity of the procedure resulting in more correctly labeled instances being filtered out.
On the other hand, it decreases as a function of $p$ due to the increased proportion of mislabeled instances available to be removed.

The sensitivity of the procedure, as discussed above, gives the proportion of mislabeled instances that are detected out of all mislabeled instances.
This measure is increases when $p$ is small and $N_1$ is large, and decreases for large $p$ and small $N_1$; corresponding to the FOR estimate.
For example, when $p = 0.20$, $N_1 = 500$, and $\rho_2 = 0.2$ the data consists of 100 mislabeled instances and 400 correctly labeled instances.
On average, the procedure removed \SI{94.5}{\percent} of mislabeled instances and only \SI{0.7}{\percent} correctly labeled instances, based on the reported sensitivity and specificity.
Thus, the resulting filtered data set has an average of 5.5 mislabeled instances and 397.2 correctly labeled instances for an average of 402.7 total instances.
This reflects a decrease in the proportion of mislabeled instances by \SI{93.15}{\percent}: while \SI{20}{\percent} of the original data set was mislabeled, only \SI{1.4}{\percent} of the filtered data set remains mislabeled.
 
\begin{table}
	\centering
	\small
	\caption{\label{tab:sim-pg0-results-n50} The $\fdr$, $\for$, sensitivity (Sens.), and specificity (Spec.) of the algorithm using simulated data with $n = 50$ and $\rho_{12} = 0.2$.}
	\begin{tabular}{rrrrrrrrrrrr}
		\toprule
		\multicolumn{1}{c}{} & \multicolumn{1}{c}{} & \multicolumn{5}{c}{$\rho_{2} = 0.2$} & \multicolumn{5}{c}{$\rho_{2} = 0.5$} \\
		\cmidrule(l{3pt}r{3pt}){3-7} \cmidrule(l{3pt}r{3pt}){8-12}
		$p$ & $N_1$ & FDR & FOR & $\%\Delta$ & Sens. & Spec. & FDR & FOR & $\%\Delta$ & Sens. & Spec.\\
		\midrule
		& 25 & 0.673 & 0.000 &  &  & 0.965 & 0.630 & 0.000 &  &  & 0.969\\
		
		& 50 & 0.993 & 0.000 &  &  & 0.933 & 0.980 & 0.000 &  &  & 0.936\\
		
		& 100 & 1.000 & 0.000 &  &  & 0.894 & 1.000 & 0.000 &  &  & 0.900\\
		
		\multirow[t]{-4}{*}{\raggedleft\arraybackslash 0.00} & 500 & 1.000 & 0.000 &  &  & 0.807 & 1.000 & 0.000 &  &  & 0.813\\
		\cmidrule{1-12}
		& 25 & 0.131 & 0.000 & 99.242 & 0.991 & 0.989 & 0.138 & 0.001 & 98.912 & 0.987 & 0.988\\
		
		& 50 & 0.310 & 0.000 & 99.305 & 0.992 & 0.975 & 0.306 & 0.000 & 99.745 & 0.997 & 0.973\\
		
		& 100 & 0.349 & 0.000 & 99.384 & 0.994 & 0.965 & 0.372 & 0.000 & 99.422 & 0.995 & 0.959\\
		
		\multirow[t]{-4}{*}{\raggedleft\arraybackslash 0.05} & 500 & 0.542 & 0.000 & 99.763 & 0.998 & 0.926 & 0.559 & 0.000 & 99.819 & 0.999 & 0.913\\
		\cmidrule{1-12}
		& 25 & 0.035 & 0.003 & 96.689 & 0.961 & 0.996 & 0.049 & 0.002 & 97.840 & 0.974 & 0.994\\
		
		& 50 & 0.046 & 0.003 & 96.583 & 0.969 & 0.994 & 0.080 & 0.002 & 97.726 & 0.980 & 0.989\\
		
		& 100 & 0.086 & 0.002 & 97.793 & 0.980 & 0.988 & 0.144 & 0.001 & 98.609 & 0.988 & 0.978\\
		
		\multirow[t]{-4}{*}{\raggedleft\arraybackslash 0.10} & 500 & 0.197 & 0.001 & 99.018 & 0.992 & 0.968 & 0.268 & 0.001 & 99.431 & 0.995 & 0.949\\
		\cmidrule{1-12}
		& 25 & 0.013 & 0.010 & 93.127 & 0.921 & 0.998 & 0.028 & 0.008 & 94.708 & 0.940 & 0.996\\
		
		& 50 & 0.016 & 0.011 & 92.580 & 0.930 & 0.997 & 0.040 & 0.007 & 95.421 & 0.957 & 0.993\\
		
		& 100 & 0.027 & 0.009 & 93.669 & 0.946 & 0.995 & 0.072 & 0.006 & 95.928 & 0.966 & 0.985\\
		
		\multirow[t]{-4}{*}{\raggedleft\arraybackslash 0.15} & 500 & 0.069 & 0.005 & 96.927 & 0.974 & 0.986 & 0.145 & 0.003 & 98.166 & 0.986 & 0.965\\
		\cmidrule{1-12}
		& 25 & 0.001 & 0.055 & 72.335 & 0.760 & 1.000 & 0.009 & 0.050 & 75.199 & 0.784 & 0.999\\
		
		& 50 & 0.006 & 0.037 & 81.397 & 0.844 & 0.999 & 0.027 & 0.030 & 84.936 & 0.874 & 0.994\\
		
		& 100 & 0.011 & 0.026 & 87.029 & 0.893 & 0.997 & 0.042 & 0.018 & 91.068 & 0.928 & 0.989\\
		
		\multirow[t]{-4}{*}{\raggedleft\arraybackslash 0.20} & 500 & 0.026 & 0.014 & 93.150 & 0.945 & 0.993 & 0.094 & 0.009 & 95.687 & 0.967 & 0.971\\
		\cmidrule{1-12}
		& 25 & 0.000 & 0.094 & 62.489 & 0.668 & 1.000 & 0.017 & 0.105 & 58.032 & 0.619 & 0.998\\
		
		& 50 & 0.003 & 0.065 & 74.092 & 0.779 & 0.999 & 0.022 & 0.060 & 75.990 & 0.793 & 0.996\\
		
		& 100 & 0.005 & 0.053 & 78.976 & 0.833 & 0.999 & 0.032 & 0.046 & 81.705 & 0.856 & 0.991\\
		
		\multirow[t]{-4}{*}{\raggedleft\arraybackslash 0.25} & 500 & 0.013 & 0.031 & 87.412 & 0.903 & 0.996 & 0.067 & 0.022 & 91.073 & 0.934 & 0.976\\
		\bottomrule
	\end{tabular}
	
\end{table}

\section{Discussion} \label{sec:discussion}
In this paper, we have presented a testing procedure for identifying incorrectly labeled instances when two or more classes are present.
Our non-parametric approach, which requires very few assumptions, yields a very high specificity, and can be implemented very easily and efficiently using standard statistical software.

As demonstrated in the simulation study, our testing procedure has a high specificity and low FOR provided the number of measurements ($n$) on each instance is large.
Decreasing the value of $n$ yields a more conservative test (i.e., one which is less likely to reject the null hypothesis in (\ref{H0i}) and remove instances), since each distance is measured less precisely.
Even for extremely small values of $n$, it was still possible to reduce the FOR and maintain a specificity over \SI{80}{\percent} for all classes.

For a fixed  $n$, classes with fewer instances (i.e. small values of $N_1$) had a higher FOR and specificity relative to larger classes.
Compared to ``classical'' classification algorithms, however, such a property is relatively unique.
Many of the existing classification procedures found in the literature have difficulties when the number of instances varies substantially across classes \citep{Sun2009}.  
This difficulty extends to those procedures utilizing these classification approaches to also detect mislabeled instances\footnote{These findings will be published in our subsequent work}.
Our testing procedure avoids this problem and maintains the integrity of small classes by analyzing each using a ``one-vs-all'' strategy  that is most conservative  for small values of $N_1$ (say for $25 \le N_1 \le 50$).

On the other hand, when the number of instances is extremely low ($2 \le N_1 \le 25$ based on our simulation studies) the accuracy of the non-parametric estimates, especially $\hat\tau$, become unreliable.
In the most extreme cases, the available data is insufficient to ever reject the null hypothesis even if  a reliable estimate of $\tau$ could be found.
For example, in LC-MS/MS proteomics, extremely small proteins ($2 \le N_1 \le 5)$ often consist entirely of inaccurate or mislabeled peptides and make up a substantial proportion of the reported proteins.
This will be addressed in a subsequent work using a complementary procedure where instances are only retained if the null hypothesis is rejected.

The use of a Bonferroni-type procedure is aimed at protecting against removing correctly identified instances is extremely conservative, prioritizing a high specificity at a cost of a higher FOR.
Even in this conservative case, the FOR in our simulations was universally reduced across all values of $N_1$, $n$, $\mathbf \rho$ and $p$.
Less conservative FWER procedures, FDR-type procedures \citep{Benjamini:1995,Benjamini2001}, or procedures seeking to explicitly control the FOR could also be considered to further reduce the FOR in the filtered data.
The primary convenience of the Bonferroni procedure is its universal applicability, especially in light of the complexity of the dependency structure of the distances and consequently of the tests statistics, $Z_i$.

Because the testing procedure estimates $\bar G$ and $\bar F$ non-parametrically using the available data, these estimates are affected by the presence of mislabeled instances. 
The resulting estimates, $\hat\tau$ and $\hat t^*$ of $\tau$ and $t^*$, may therefore be biased when mislabeled instances are included in the class.
One possible method of remedy is to iteratively remove a small number of instances and re-estimate $\tau$ and $t^*$ until some stopping criterion is met.
Developing such a sequential estimation procedure is left to future work.

Although we have used correlation as a measure distance, the procedure is generally applicable whenever the observed data for each instance can be effectively combined as a "measure of distance".
In our subsequent work, we also include a demonstration of the procedure using Manhattan distance.

\section*{Acknowledgement(s)}

M. Key would like to thank Dr. Susanne Ragg for her assistance and support for this work.

\section*{Disclosure statement}

N.A.

\bibliographystyle{tfs}
\bibliography{references}

\appendix
\section{Proof of Lemma \ref{lem:aControlsb}}\label{app:lemma2}

\aControlsb*
\begin{proof}

	Let $b(k, p, n) = \sum_{j = 0}^{k} {n \choose j} p^j(1 - p)^{n-j})$ be the binomial sum up to $k < n$.
	This is clearly a polynomial of degree $n$, and thus continuous on $\R$, and specifically on $(0,1)$.
	For $\tau > \tau^*$, let $Z \sim Bin\left(\frac{N_1-1}{2}, \tau\right)$ and set $X \sim Bin\left(\frac{N_1-1}{2}, \tau^*\right)$.
	Because $\tau > \tau^*$, 
	\begin{equation*}
	\Pr\left(Z \le \frac{N_1-1}{2}\middle| \tau\right) < \Pr\left(X \le \frac{N_1-1}{2}\middle |\tau^*\right) = \Pr(Z \le a_\alpha | \tau),
	\end{equation*}
	 by \citep{Gilat1977}, so $\Pr\left(Z \le \frac{N_1-1}{2}\middle| \tau\right) < \Pr(Z \le a_\alpha | \tau)$ and thus $a_\alpha >  \frac{N_1-1}{2}$.\\[.5ex]
	
	Since $a_\alpha \ge \frac{N_1 - 1}{2}$, then $N_1 - 1 - a_\alpha \le \frac{N_1 - 1}{2}$.
	Consequently,
	\begin{align*}
	\beta &= \Pr(Z_i \ge a_\alpha | 1 - \tau) \\
	& =  \Pr(Z_i \le N_1 - 1 - a_\alpha | \tau) \\
	& \le \Pr(Z_i \le a_\alpha | \tau) \\ &	 = \alpha.
	\end{align*}
	\end{proof}	
\end{document}